\documentclass[a4paper,12pt]{article}

\usepackage[T1]{fontenc}
\usepackage{lmodern}
\usepackage[utf8]{inputenc}
\usepackage[backend=biber,style=alphabetic,maxbibnames=100,maxalphanames=4]{biblatex}
\usepackage[in]{fullpage}
\usepackage[dvipsnames,svgnames,x11names]{xcolor}
\usepackage[colorlinks]{hyperref}
\usepackage{enumerate}
\usepackage{nicefrac}
\usepackage{subfigure}
\usepackage{graphicx}
\usepackage{algorithm}
\usepackage{amsmath}
\usepackage{amsfonts,amssymb,amsthm,bbm,mathtools}
\usepackage{mathtools}
\usepackage{authblk}
\usepackage{amssymb}
\usepackage{amsthm}
\usepackage{esvect}
\usepackage[normalem]{ulem}
\usepackage{braket}
\usepackage{pgfplots}
\usepackage{multicol}
\usepackage[final]{microtype}
\usepackage[font=small]{caption}
\usetikzlibrary{calc}
\usetikzlibrary{trees}

\makeatletter
\newtheorem*{rep@theorem}{\rep@title}
\newcommand{\newreptheorem}[2]{%
\newenvironment{rep#1}[1]{%
 \def\rep@title{#2 \ref{##1}}%
 \begin{rep@theorem}}%
 {\end{rep@theorem}}}
\makeatother
\newreptheorem{theorem}{Theorem}

\pgfplotsset{compat=1.6}

\bibliography{bibliography}

\hypersetup{
    breaklinks = true,
	colorlinks = true,
	citecolor = teal,
	urlcolor = purple,
}

\newtheorem{theorem}{Theorem}

\newtheorem{claim}[theorem]{Claim}

\newcommand{\Q}{\mathrm{Q}}
\newcommand{\E}{\mathrm{e}}
\newcommand{\I}{\mathrm{i}}

\DeclareMathOperator{\Oh}{O}
\DeclareMathOperator{\oh}{o}

\DeclareMathOperator{\poly}{poly}
\DeclareMathOperator{\aux}{aux}

\newcommand{\Mod}[1]{\ (\mathrm{mod}\ #1)}

\newcommand{\norm}[1]{\left\lVert#1\right\rVert}

\title{Quantum Search on Computation Trees}
\author{Jevgēnijs Vihrovs}

\affil{Centre for Quantum Computer Science, Faculty of Science and Technology,\authorcr University of Latvia, Raiņa 19, Riga, Latvia, LV-1050}

\date{}

\begin{document}

\maketitle

\begin{abstract}
    We show a simple generalization of the quantum walk algorithm for search in backtracking trees by Montanaro (ToC 2018) to the case where vertices can have different times of computation.
    If a vertex $v$ in a tree of depth $D$ is computed in $t_v$ steps from its parent, then we show that detection of a marked vertex requires $\Oh(\sqrt{TD})$ queries to the steps of the computing procedures, where $T = \sum_v t_v^2$.
    This framework provides an easy and convenient way to re-obtain a number of other quantum frameworks like variable time search, quantum divide \& conquer and bomb query algorithms.
    The underlying algorithm is simple, explicitly constructed, and has low poly-logarithmic factors in the complexity.
    
    As a corollary, this gives a quantum algorithm for variable time search with unknown times with optimal query complexity $\Oh(\sqrt{T \log \min(n,t_{\max})})$, where $T = \sum_i t_i^2$ and $t_{\max} = \max_i t_i$ if $t_i$ is the number of steps required to compute the $i$-th variable.
    This resolves the open question of the query complexity of variable time search, as a matching lower bound was recently shown by Ambainis, Kokainis and Vihrovs (TQC'23).
    As another result, we obtain an $\widetilde \Oh(n)$ time algorithm for the geometric task of determining if any three lines among $n$ given intersect at the same point, improving the $\Oh(n^{1+\oh(1)})$ algorithm of Ambainis and Larka (TQC'20).
\end{abstract}

\section{Introduction}

For many algorithms, the underlying computation process can be naturally described as exploring a tree until finding a solution at one of the leaves.
Just some of the examples include backtracking procedures, divide \& conquer algorithms and search data structures.
In \cite{Montanaro15}, Montanaro gave a quantum walk algorithm for searching for a solution in classical backtracking trees, with application to constraint satisfaction problems.
In general, their algorithm can be applied for detection of a solution in a tree of size $T$ and depth $D$, using $\Oh(\sqrt{TD})$ steps of a quantum walk.

In this work we examine a generalization of this problem of detecting a solution (a marked vertex) in a rooted tree, which might not be known beforehand.
Each child can be computed locally from its parent by a transition procedure; we consider that computation of each vertex $v$ from its parent requires $t_v$ transition steps.
By adjusting the quantum walk in Montanaro's algorithm, we show the following result:
\begin{theorem}[informal] \label{thm:main}
    Suppose that the costs $t_v$ are known beforehand and the depth of the tree is at most $D$.
    Then there exists a bounded-error quantum algorithm that detects a marked vertex using $\Oh(\sqrt{TD})$ queries to the individual steps of the transition procedures, where $T = \sum_v t_v^2$.
    If $t_v$ are unknown, the algorithm has query complexity $\Oh(\sqrt{TD\log t_{\max}})$, where $t_{\max} = \max_v t_v$.
\end{theorem}

The condition that the times are known beforehand can appear, for example, in divide an conquer algorithms, where the computation tree is a full $k$-ary tree and each level of recursion has an upper bound on its running time.
The unknown times case is more natural for backtracking trees, in which case the complexity gets worse just by a square root of the logarithm.
The dependence on $T$ in both cases is optimal, since this problem is a generalization of quantum search with variable times, whose query complexity is essentially characterized by $\sqrt{T}$ \cite{Amb10}.

While the results we obtain using the algorithm are generally well-known, our main focus is on the simplicity and convenience of application.
The algorithm has many useful qualities -- first of all, it is a simple and natural quantum walk algorithm.
It has very low and in some cases optimal poly-logarithmic factors for query complexity.
This approach also gives a useful general framework in which we can plug in existing algorithms to produce a quantum speedup with a provable complexity guarantee.

\subsection{Our results}

We apply Theorem \ref{thm:main} to reproduce three quantum frameworks where the respective computation process can be seen as a search for marked vertices in a tree with different computation times for each vertex (a \emph{computation} tree).
Each of them is identified by a class of trees as shown in Figure \ref{fig:trees}.

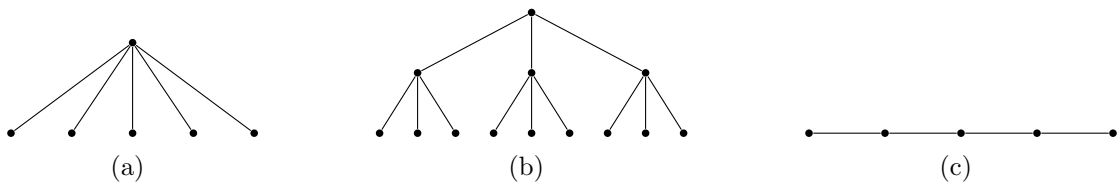
\begin{figure}[H]
    \centering
    \subfigure[]{
        \begin{tikzpicture}[
          level 1/.style={sibling distance=.8cm, level distance=1.2cm},
          vertex/.style={circle, fill=black, inner sep=1pt}
        ]
          \node[vertex] {}
            child {node[vertex] {} }
            child {node[vertex] {} }
            child {node[vertex] {} }
            child {node[vertex] {} }
            child {node[vertex] {} };
        \end{tikzpicture}
    }
    \hspace{1cm}
    \subfigure[]{
        \begin{tikzpicture}[
          level 1/.style={sibling distance=1.5cm, level distance=.8cm},
          level 2/.style={sibling distance=.5cm, level distance=.8cm},
          vertex/.style={circle, fill=black, inner sep=1pt}
        ]
          \node[vertex] {}
            child {node[vertex] {} 
                child {node[vertex] {} }
                child {node[vertex] {} }
                child {node[vertex] {} }
            }
            child {node[vertex] {} 
                child {node[vertex] {} }
                child {node[vertex] {} }
                child {node[vertex] {} }
            }
            child {node[vertex] {} 
                child {node[vertex] {} }
                child {node[vertex] {} }
                child {node[vertex] {} }
            };
        \end{tikzpicture}
    }
    \hspace{1cm}
    \subfigure[]{
    \begin{tikzpicture}[
        vertex/.style={circle, fill=black, inner sep=1pt},
        dummy/.style={circle, fill=white}]
      \node[vertex] (v1) at (0,0) {};
      \node[vertex] (v2) at (1,0) {};
      \node[vertex] (v3) at (2,0) {};
      \node[vertex] (v4) at (3,0) {};
      \node[vertex] (v5) at (4,0) {};
      
    
      \draw (v1) -- node[above] {} (v2);
      \draw (v2) -- node[above] {} (v3);
      \draw (v3) -- node[above] {} (v4);
      \draw (v4) -- node[above] {} (v5);
    \end{tikzpicture}
    }
    \caption{A way of seeing different quantum frameworks as computation trees: (a) a star in the case of variable time search; (b) the divide \& conquer process is represented by a full $k$-ary tree; (c) a line in the case of quantum bomb query algorithms.} \label{fig:trees}
\end{figure}

\paragraph{Search with variable times.}

First we apply our algorithm to the problem of quantum variable time search (VTS), which is a generalization of unstructured search solved by Grover's algorithm.
In this problem, the $i$-th item is computed by a procedure $A_i$ that takes $t_i$ steps, with different values for separate items.
Such situations arise naturally when a search problem is split into multiple subproblems of different size, or when we compose independent procedures with \textsc{Or}, see for example \cite{Amb12, LeG14, CKS17}.

In this work we primarily focus on the optimal quantum \textit{query complexity} of VTS.
Here, by a single query we mean an application of a single step of $A_i$ for a superposition of indices $i$.
In \cite{Amb10}, Ambainis introduced and analyzed this problem, determining that the complexity is essentially characterized by $\sqrt T$, where $T := t_1^2+\ldots+t_n^2$.
They in fact proved that if all $t_i$ are known beforehand, then the quantum query complexity is exactly characterized by $\Theta(\sqrt T)$.

For the other case, when $t_i$ are not known beforehand, they gave a quantum algorithm with query complexity $\Oh(\sqrt{T} \log n \log \log n \log t_{\max} \log \log t_{\max})$, where $t_{\max} := \max_i t_i$.
Even though intuitively this is a harder setting, a lower bound stronger than $\Omega(\sqrt T)$ was lacking.
In \cite{Amb12}, quantum amplitude amplification was generalized to boost the success probability of quantum algorithms with variable stopping time.
For the original VTS problem, it implied an $\Oh(\sqrt T \log^{1.5} t_{\max})$ quantum query algorithm, which shaved off some of the logarithmic factors.

More recently some even more refined algorithms were given.
First, using multi-dimensional quantum walks, Jeffery generalized the quantum walk framework to the setting where transitions on edges take variable times \cite{Jef25}.
As a corollary, they obtained a quantum algorithm with \emph{time complexity} $\Oh(\sqrt{T} \log^{1.5} t_{\max})$.\footnote{We note, however, that different works use slightly different complexity models. For example, if $t_{\max}=1$, then we would get an $\Oh(\sqrt n)$-time algorithm for Grover's search, without any logarithmic factors, which is not known to be possible. This happens here because of the assumption that preparing a uniform superposition over $n$ items takes unit cost.}
Ambainis, Kokainis and Vihrovs gave a simple $\Oh(\sqrt{T} \log n)$ quantum query algorithm, and also finally proved a query lower bound of $\Omega(\sqrt{T \log T})$ \cite{AKV23}.
Even more recently, Cornelissen gave a quantum query algorithm that solves hard instances of this lower bound in the same complexity $\Oh(\sqrt{T \log T})$ \cite{Cor25}.

Regarding Jeffery's algorithm, we observe that their VTS algorithm already has query complexity $\Oh(\sqrt{T \log t_{\max}})$! -- because a factor of $\Oh(\log t_{\max})$ comes from the circuit implementation of a step of their quantum walk, but a single step makes only a constant number of queries to the subprocedures.
This closes the gap of a square root of a logarithm between upper and lower bounds of VTS with unknown times.

In this work, by applying Theorem \ref{thm:main} to VTS with unknown times (Theorem \ref{thm:vts}), we obtain an even simpler quantum walk algorithm also with query complexity $\Oh(\sqrt{T \log t_{\max}})$.
If $n < t_{\max}$, then by grouping every $t_{\max}/n$ steps of procedures into a single step, we get complexity $\Oh(\sqrt{T \log n})$.
The lower bound of \cite{AKV23} can be adjusted to be tight in both cases depending on which of $n$ or $t_{\max}$ is larger.
Therefore, this resolves the open question of the precise asymptotic query complexity of VTS with unknown times.
Our algorithm and the ones of \cite{Jef25, Cor25} share close similarities, as all can be seen as applications of the electrical network framework \cite{Bel13,BCJKM13} with similar weighting schemes. 
\begin{reptheorem}{thm:vts}
    Suppose that we are given upper bound estimates $T'$ and $t'$ such that $T \leq T' = \Oh(T)$ and $t_{\max} \leq t' = \poly(t_{\max})$.
    Then the quantum query complexity of VTS with unknown times is 
    \[\Theta(\sqrt{T \log \min(n, t_{\max})}).\]
\end{reptheorem}

There are a couple of details to note in comparison with previous algorithms.
Firstly, the algorithm requires upper bound estimates on both $T$ and $t_{\max}$, which some of the previous algorithms do not \cite{Amb10}, but others also do \cite{AKV23, Jef25}.
However, this assumption is quite mild, as most commonly VTS is applied in cases where we obtain explicit complexity upper bounds and both $T$ and $t_{\max}$ are known beforehand from the complexity analysis.

Secondly, we assume that the algorithm has query access to the individual steps of the subprocedures $A_i$.
This is in contrast to algorithms \cite{Amb10, AKV23}, where the subprocedures act as black-box oracles: more specifically, any procedure $A_i$ can be run for a chosen amount of $t$ steps and the output is $\star$ if $t < t_i$ (meaning that the computation has not finished) and $x_i$ if $t \geq x_i$.
This assumption is also not very restrictive for the design of quantum query algorithms, since usually we know the subprocedures at hand and their descriptions (e.g.~a list of gates).
Jeffery's algorithm \cite{Jef25} also assumes this condition, in fact it assumes QRAM access to the gates of the subprocedures.
In general, some type of QRAM is necessary for general case VTS algorithms to call any type of $A_i$ oracles in superposition, if we are concerned with the quantum time complexity.

\paragraph{Divide \& conquer algorithms.}

Our second application is to the quantum version of the divide \& conquer paradigm.
In the typical classical divide \& conquer formulation, a problem of size $n$ is partitioned into $a$ subproblems of size $n/b$, which are solved recursively (we assume that $a$ and $b$ are constants).
The results are then combined by some function and passed back to the upper recursion call.
The additional computation cost incurred by the partitioning and combining is denoted by $T_{\aux}(n)$.
The classical cost of the algorithm then can be computed by solving a recurrence
\[T(n) = aT(n/b) + T_{\aux}(n).\]

Childs et al.~adapted this framework for quantum query complexity and showed a quantum speedup for certain cases of the combining functions \cite{CKKSW25}.
For example, if combining is performed by \textsc{Or} or \textsc{And}, the quantum query complexity of the corresponding problem is upper bounded by $\Oh(T_{\Q}(n))$, where
\[T_{\Q}(n) \leq \sqrt{a}T_{\Q}(n/b) + T_{\aux}(n),\]
which in many cases can lead to a quadratic speedup over the classical algorithm.

In \cite{ABBLS25}, time-efficient quantum algorithms were proposed for divide \& conquer.
Their main technical result is proving that Grover's algorithm can accept $N$ bounded-error inputs and still have constant success probability using only $\Oh(\sqrt N \log N)$ gates, without the need for $\Oh(\log N)$ repetitions for error reduction.
They apply their algorithm to several forms of divide \& conquer, obtaining time-efficient quantum algorithms, and in some cases optimal quantum query algorithms.

For recursive versions of divide \& conquer where combining is performed by disjunction or minimization, they apply their modification of Grover's algorithm recursively (which they call constructible instances).
Since there is no need for error reduction at each recursive level, this allows to cut a factor of $\Oh((\log \log n)^{\log n})$.
However, since there is still some constant factor $c$ in the complexity of Grover's search and the depth of the recursion is $\Oh(\log n)$, they still have an additional factor of $\Oh(c^{\log n}) = \poly(n)$.

Returning to the multi-dimensional quantum walks, Jeffery and Pass gave another time-efficient algorithm for quantum divide \& conquer \cite{JP24}.
They showed a quantum decision algorithm with time complexity $\Oh(T_{\Q}(n) \log T_{\Q}(n))$ for cases when the combining function is a Boolean symmetric function.
The results of \cite{ABBLS25} and \cite{JP24} are incomparable, the former for example can solve minimization problems, while the latter can tackle any symmetric formula in a combining step.

Theorem \ref{thm:main} implies an analogous result in terms of query and time complexity:
\begin{reptheorem}{thm:dnq}
    There is a bounded-error quantum algorithm for divide \& conquer with disjunctive combining with query complexity $\Oh(T_{\Q}(n) \sqrt{\log n})$ and time complexity $\widetilde \Oh(T_{\Q}(n))$ in the QRAM model.
\end{reptheorem}
On one hand, the query complexity of this result is by a $\sqrt{\log n}$ factor weaker than in \cite{CKKSW25}.
On the other hand, the algorithm can be applied to speed up classical algorithms in RAM model.
Our algorithm also avoids the $O(c^{\log n})$ factor for the disjunctive combining compared to \cite{ABBLS25}.

We apply our result to the \textsc{Point-On-3-Lines} geometric problem: given $n$ lines in the plane, are there $3$ lines that intersect at the same point?
This problem is a classic representative of the \textsc{3-Sum-Hard} class of problems and has been studied in the context of fine-grained complexity.
While classically its complexity is $\widetilde \Theta (n^2)$ assuming the hardness of \textsc{3-Sum} \cite{GO95}, quantumly Ambainis and Larka showed an $\Oh(n^{1+\oh(1)})$ quantum algorithm \cite{AL20}, while Buhrman et al.~proved a linear lower bound conditional on the quantum version of \textsc{3-Sum} hardness \cite{BLPS22a}.
By using Theorem \ref{thm:dnq} together with geometric $(1/r)$-cuttings \cite{Cha93,CT16}, we give a refinement of their result:
\begin{reptheorem}{thm:po3l}
    There is a bounded-error quantum algorithm that solves \textsc{Point-On-3-Lines} with time complexity $\widetilde \Oh(n)$ in the QRAM model.
\end{reptheorem}

\subsection{Related work}

Our algorithm can be viewed as a combination of ideas of Jeffery's variable-time quantum walks \cite{Jef25} with Montanaro's quantum backtracking algorithm \cite{Montanaro15}.
Both of these algorithms utilize Belov's electrical network quantum walk framework \cite{Bel13, BCJKM13}.
The complexity of $\Oh(\sqrt{TD})$ in Theorem \ref{thm:main} essentially comes from the fact that in the appropriately weighted tree, the effective resistance $\mathcal R$ is $\Oh(D)$ and the total weight $\mathcal W$ of the edges of the tree is $\Oh(T)$, and the complexity then can be expressed as $\Oh(\sqrt{\mathcal R\mathcal W})$.
The result of \cite{Jef25} is very general and adapts quantum walk frameworks on arbitrary graphs to the variable times setting, also incorporating procedures with transition times being random variables.

The approach in this work is much simpler because we consider only trees and exact transition procedures with concrete running times, and the analysis of the complexity is also very short and requires only ``elementary'' quantum techniques.
This algorithm is able to compute the description of the child vertices from the description of their parents because in a tree we can store all of the (unique) computation on the path that leads to the marked vertex, while this is not possible in a general graph (as we need to uncompute the working space to allow the quantum branches from multiple paths to interfere).

\section{Preliminaries}

\subsection{Notation}

We use the notation $f(n) \sim g(n)$ when $f(n) = (1+\oh(1)) g(n)$.
When $a$ is proportional to $b$, we use the notation $a \propto b$.
We also use the notation $\mathbbm 1[A]$ for the indicator function of $A$ that is $1$ or $0$ if $A$ is true or false, respectively.

\subsection{Search in a tree}

First we give an abstract description of the problem.
Our goal is to find a \emph{marked} vertex in a rooted tree, where the children can be computed locally from their parents.
The marked vertices correspond to the solutions in this computation tree.
Since in such trees the search stops after finding a solution, we assume that marked vertices can be only the leaves of the tree.

Let $G=(V,E)$ be a rooted tree graph with the set of vertices in $V$ and edges in $E \subset V \times V$.
If a vertex $v$ is a child of a vertex $u$ in $G$, we denote such a relationship by $u \underset{G}{\to} v$.
The number of children of a vertex $u$ we denote by $\deg(u) = |\{u \mid u \underset{G}{\to} v\}|$.
Note that $\deg(u) = 0$ iff $u$ is a leaf.
For a vertex $v$, we denote by $u \underset{G}{\rightsquigarrow} v$ that a vertex $u$ is on the path from the root to $v$ in $G$ (including the root).
We denote the size of $G$ by the number of its vertices, $|G|:=|V|$, and we also use notation $v \in G$ to denote $v \in V$.

Firstly, we assume that there is a unitary oracle that tells us the number of children of a vertex $v$,
\[C\ket{v}\ket{0} \mapsto \ket{v} \ket{\deg(v)}.\]
We also have a unitary oracle that marks the solutions,
\[M\ket{v}\ket{0} \mapsto \ket{v} \ket{\mathbbm 1[\text{$v$ is marked}]}.\]

We say that a rooted tree $\mathcal T$ is a \emph{computation} tree if for a vertex $v$, we need $t_v$ steps (classical or quantum) to compute its description from the description of its parent; if $r$ is the root of $\mathcal T$, then we assume that $t_r = 0$ and $t_v > 0$ for $v \neq r$.
We will denote by $n := |\mathcal T|-1$ the number of non-root vertices of $\mathcal T$.
The complexity of our algorithms will depend on $T := \sum_{v \in \mathcal T} t_v^2$, and on the depth of $\mathcal T$ which we denote by $D$.

We distinguish two settings: \textit{known} and \textit{unknown} computation times.
In the known times setting, we assume that we know how many steps it takes to compute a child $v$ before we even start to compute its description, given only the description of its parent $u$ and the index of the child $i \in [\deg(u)]$.
In other words, we assume that we have a unitary oracle
\[O_t \ket{u} \ket{i} \ket{0} \mapsto \ket{u} \ket{i} \ket{t_v}.\]

We also abstract a step of a transition procedure by another unitary oracle.
Denote by $\ket{v_t}$ the working space of $\mathcal T$ after performing $t$ steps of transition to $v$.
Note that we have $\ket{v_0} = \ket{0}$ and $\ket{v_{t_v}} = \ket{v}$.
We denote the computation oracle by $A$ and the application of one step is given by
\[A \ket{u} \ket{i} \ket{v_t} \ket{t} \mapsto \ket{u} \ket{i} \ket{v_{t+1}} \ket{t+1} .\]

The steps do not necessarily have to be single gates of a quantum procedure; they can be abstracted as necessary (for example, a step could be a longer sequence of gates).
However, we should be able to apply any of these oracles in superposition.

In the unknown times, we only know whether the computation is done once it is actually finished.
That is, we have an oracle
\[O_f \ket{u} \ket{i} \ket{v_t} \ket{t} \ket{0} \mapsto \ket{u} \ket{i}  \ket{v_t} \ket{t} \ket{\mathbbm 1[t \geq t_v]}.\]

\subsection{Model}

We distinguish two complexities of the algorithm.
The first is the query complexity, the number of calls to the oracle $A$.
The second is the gate (or time) complexity, in which we use the standard quantum circuit model together with Quantum Random Access Gates (QRAG) (see, for example, \cite{ABDLS24}).
A QRAG gate implements the following mapping:
\[\ket{i}\ket{b}\ket{x_1,\ldots,x_N} \mapsto \ket{i}\ket{x_i}\ket{x_1,\ldots,x_{i-1},b,x_{i+1},\ldots,x_N}.\]
Here, the last register represents the memory space of $N$ bits.
Essentially, QRAG gates allow both for reading and writing memory operations in superposition.
The QRAG model is an analogue of the classical RAM model, and we will use the fact that any classical time $T$ algorithm in RAM can be implemented in $\Oh(T)$ time in the QRAG model.
Below we will just refer to our model as the QRAM model.

\subsection{Tools}

Similarly to \cite{Montanaro15}, we require a couple of tools from quantum walk analysis.

\begin{theorem}[Effective Spectral Gap Lemma \cite{SSRTM11}] \label{thm:effective-lemma}
    Let $\Pi_A$ and $\Pi_B$ be projectors on the same Hilbert space, and let $R_A = 2\Pi_A - I$, $R_B = 2\Pi_B - I$.
    Let $P_{\epsilon}$ be the projector onto the span of the eigenvectors of $R_B R_A$ with eigenvalues $\E^{2\I\theta}$ such that $|\theta| \leq \epsilon$.
    Then, for any vector $\ket{\psi}$ such that $\Pi_A\ket{\psi} = 0$, we have
    \[\norm{P_{\epsilon}\Pi_B \ket{\psi}} \leq \epsilon \norm{\ket{\psi}}.\]
\end{theorem}



We also need the quantum phase estimation where we'll only need to distinguish eigenvalue $1$ from other eigenvalues.
This formulation follows, for example, from \cite{MNRS11,Montanaro15} by picking sufficiently large $s = \Oh(\log(1/\epsilon))$.

\begin{theorem}[Quantum Phase Estimation \cite{Kit95,CEMM98}] \label{thm:phase-estimation}
    For every unitary $U$ and a precision parameter $\epsilon > 0$, there exists a quantum circuit $C$ that 
    \begin{itemize}
        \item uses $\Oh(1/\epsilon)$ controlled-$U$ applications and $s = \Oh(\log^2(1/\epsilon))$ other gates;
        \item if $\ket{\psi}$ is an eigenvector of $U$, then
        $C\ket{\psi}\ket{0^s} = \ket{\psi}\ket{\omega}$ for some $\ket{\omega}$; if $\ket{\psi}$ is an eigenvector with eigenvalue $1$, then $\ket{\omega} = \ket{0^s}$;
        \item for any state $\ket{\phi} = \sum_k \alpha_k \ket{\psi_k}$, where $\ket{\psi_k}$ is an eigenvalue of $U$ with eigenvalue $\E^{2\I\theta_k}$,
        \[C\ket{\phi}\ket{0^s} = \sum_k \alpha_k \ket{\psi_k} \ket{\omega_k},\]
        where $\sum_{k : \theta_k \geq \epsilon} \left|\braket{\omega_k|0^s}\right|^2 \leq 1/8$.
    \end{itemize}
\end{theorem}

\section{Quantum walk}

In this section first we describe the quantum-walk based marked vertex detection algorithm, following Montanaro's description.
First we describe the algorithm for the setting with known transition times; the unknown time case can be simply reduced to known times, as we will see later.

\subsection{Known times}

Recall that the transition procedure $A$, when computing $v$, goes successively through the states $\ket{v_1}$, $\ldots$, $\ket{v_{t_v}}$.
Suppose that $u$ is the parent of $v$ in $\mathcal T$; we expand our tree by adding $t_v-1$ intermediate vertices in a path between $u$ and $v$.
We will denote the respective vertices by $s(v,t)$ for $t \in [t_v]$, each associated with the respective step of computation of $v$.
We call this the \emph{expanded tree} and denote it by $\mathcal E$.
For simplicity, we will still denote the root of $\mathcal E$ by $r$.
The marked vertices in the expanded tree are those $s(v,t_v)$ such that $v$ is marked in $\mathcal T$.

To define the unitary of the quantum walk, we assign a \emph{weight} function to each vertex $x \in \mathcal E$:
\[w_x = \begin{cases} \nicefrac{1}{D}, & \text{if $x = r$,} \\ t_v, & \text{if $x = s(v,t)$ for some $t$.} \end{cases}\]

Next, for each vertex $x \in \mathcal E$, we define the diffusion operator $D_x$ which acts on the subspace $\mathcal H_x$ spanned by $\ket{x}$ and its children $\ket{y}$:
\begin{itemize}
    \item if $x$ is marked, $D_x = I$;
    \item otherwise, $D_x = I-2\ket{\psi_x}\bra{\psi_x}$, where $\ket{\psi_x}$ is a unit vector defined by
 \[\ket{\psi_x} \propto \sqrt{w_x}\ket{x} + \sum_{y\,:\,x \underset{\mathcal E}{\rightarrow} y} \sqrt{w_y} \ket{y}.\]
\end{itemize}

To define the unitary operator for the steps of the quantum walk, let $A$ and $B$ be the sets of vertices of $\mathcal E$ at even and odd distance from the root, respectively.
A single step of the quantum walks applies $U = R_BR_A$, where
\[
    R_A = \bigoplus_{x \in A} D_x, \hspace{1cm}
    R_B =  \ket{r}\bra{r} + \bigoplus_{x \in B} D_x.
\]
In other words, the quantum walk alternates between applying diffusion at vertices at even and odd distance from the root.
The starting state of the walk is the root state $\ket{r}$.

Now that the unitary of the step of the quantum walk is defined, the overall quantum algorithm is as follows.
\begin{algorithm}[H] 
    \caption{Marked vertex detection algorithm.}\label{algo}

    \begin{enumerate}
        \item Apply Quantum Phase Estimation (Theorem \ref{thm:phase-estimation}) on $U$ and the state $\ket{r}$ with precision $\epsilon = \Oh(1 / \sqrt{TD})$ and measure the phase register.
        \item Accept if the measured value is $0^s$ and reject otherwise.
    \end{enumerate}
\end{algorithm}

We first examine the running time of this abstracted quantum walk algorithm and afterwards we turn to its circuit implementation for computation trees.

\begin{theorem} \label{thm:known}
    For a computation tree $\mathcal T$ with depth $D$ and known transition times $t_v$, quantum Algorithm \ref{algo} detects the presence of a marked vertex with constant success probability and uses $\Oh(\sqrt{TD})$ applications of $U$, where $T = \sum_v t_v^2$.
\end{theorem}

\begin{proof}
First suppose there is a marked leaf $v$ in the tree $\mathcal T$, then $m := s(v,t_v)$ is marked in the expanded tree $\mathcal E$.
Examine the vector
\[\ket{\phi_m} = \sum_{x \underset{\mathcal E}{\rightsquigarrow} m} \frac{(-1)^{\ell(x)}}{\sqrt{w_x}} \ket{x},\]
where $\ell(x)$ is the distance of $x$ from the root.
First we show that $\ket{\phi_m}$ is a $1$-eigenvector of $R_A$.
Recall that $R_A$ is the direct sum of operators $D_x$ that act on subspaces $\mathcal H_x$ for $x \in A$, and these subspaces do not overlap.
It is sufficient to show that $D_x\Pi_{\mathcal H_x}\ket{\phi_m}=\Pi_{\mathcal H_x}\ket{\phi_m}$.

If $\mathcal H_x$ does not overlap with $\ket{\phi_m}$, then $D_x$ does not affect $\ket{\phi_m}$.
Otherwise, if $x$ is not a leaf, then $\mathcal H_x$ overlaps with  $2$ vertices in the support of $\ket{\phi_m}$.
We then have
\[\braket{\phi_m|\psi_x} \propto \frac{(-1)^{\ell(x)}}{\sqrt{w_x}}\cdot \sqrt{w_x} + \frac{(-1)^{\ell(x)+1}}{\sqrt{w_y}}\cdot \sqrt{w_y} = 0,\]
where $y$ is the child of $x$ on the path $x \underset{\mathcal E}{\rightsquigarrow} m$.
Since $D_x = I-2\ket{\psi_x}\bra{\psi_x}$, it acts as identity on $\Pi_{\mathcal H_x} \ket{\phi_m}$.
Finally, if $x = m$, we have $D_x = I$.
Similarly, $\ket{\phi_m}$ is also a $1$-eigenvector of $R_B$ (with $\ket{r}\bra{r}$ acting as identity on $\ket{r}$).

We then have
\[\norm{\ket{\phi_m}}^2 = \sum_{x \underset{\mathcal E}{\rightsquigarrow} m} \frac{1}{w_x} = D + \sum_{u \underset{\mathcal T}{\rightsquigarrow} v, u \neq r} t_u \cdot \frac{1}{t_u} \leq 2D\]
and
\[\frac{\left|\braket{r|\phi_m}\right|}{\norm{\ket{\phi_m}}} \geq \frac{1}{\sqrt{2}}.\]
Hence, the algorithm accepts with probability at least $1/2$ regardless of the precision.

Now consider the case when there are no marked vertices; examine the vector
\[\ket{\eta} = \sum_{x \in \mathcal E} \sqrt{\frac{w_x}{w_r}} \ket{x}.\]
Let $\Pi_A$ and $\Pi_B$ be the projectors onto the invariant subspaces of $R_A$ and $R_B$.
It is evident that $\ket{\eta}$ is proportional to $\ket{\psi_x}$ for each $x \in \mathcal E$.
Since each $D_x$ negates the state $\ket{\psi_x}$, and subspaces $\mathcal H_x$ do not overlap for $x \in A$, we have $\Pi_A \ket{\eta} = 0$.
Similarly, $\Pi_B \ket{\eta} = \sqrt{\frac{w_r}{w_r}} \ket{r} = \ket{r}$.
Denote by $P_\epsilon$ the projector onto the span of eigenvectors of $U$ with eigenvalues $\E^{2\I\theta}$ such that $|\theta| \leq \epsilon$.
Then it follows from Theorem \ref{thm:effective-lemma} that
\begin{align*}
\norm{P_\epsilon \ket{r}} 
&= \norm{P_\epsilon \Pi_B \ket{\eta}}
\leq \epsilon \norm{\ket{\eta}} 
= \epsilon \sqrt{\sum_{x \in \mathcal E} \frac{w_x}{w_r}}
= \epsilon \sqrt{1 + D \sum_{\substack{x \in \mathcal E, \, x \neq r}} w_x}
 = \epsilon \sqrt{1 + D \sum_{v \in \mathcal T} t_v^2}.
\end{align*}
The last step holds because for each $v \in \mathcal T$, $v \neq r$, there are $t_v$ vertices $v_1$, $\ldots$, $v_{t_v}$ with weight $t_v$.
Therefore, for phase estimation we can choose small enough $\epsilon = \Omega(1/\sqrt{T D})$ so that $\norm{P_\epsilon \ket{r}}^2 \leq 1/8$.
Then the total probability of measuring $0^s$ is small:
\[\sum_{k : \theta_k \leq \epsilon} \left|\braket{\omega_k|0^s}\right|^2 + \sum_{k : \theta_k \geq \epsilon} \left|\braket{\omega_k|0^s}\right|^2 \leq \norm{P_\epsilon \ket{r}}^2 + \sum_{k : \theta_k \geq \epsilon} \left|\braket{\omega_k|0^s}\right|^2 \leq 1/4.\]

We can see that $U$ is applied $\Oh(\sqrt{TD})$ times in phase estimation, plus $\Oh(\log^2(TD))$ other gates.
By repeating the phase estimation a constant number of times, we can distinguish between the marked and unmarked case by looking at the statistics, with constant error.
\end{proof}

\subsection{Implementation}

\paragraph{State of the algorithm.}
At this point we have defined the quantum walk in terms of the operator $U$, which we will now proceed to implement.
We begin by describing the overall quantum state of the algorithm.

First, we note that in the proof above we used abstract labels $s(v,t)$ for convenience, but haven't defined them yet.
We will encode $\ket{s(v,t)}$ by a state that stores the \textit{whole} computational history on the path to computing the description of $v_t$, starting from the root $r$.
Consider the original $\mathcal T$ algorithm at a point of having computed the vertices $v^1$, $v^2$, $\ldots$, $v^{k-1}$ and $t > 0$ steps of $v^k = v$.
Then $\ket{s(v,t)}$ will contain all $\ket{v^1}$, $\ldots$, $\ket{v^{k-1}}$, $\ket{v^k_t}$, with some additional registers.

Now we describe the state of the algorithm more formally.
We note that $R_A$ and $R_B$ are quite similar; first we have a single qubit register $\ket{p} \in \mathcal H_p$, storing the parity of the quantum walk step.
If it is equal to $0$, then $R_A$ should be applied; otherwise $R_B$ should be applied.
By encoding the parity in a separate qubit, we can then implement $R_A$ and $R_B$ by the same transformation, as we will see later.
Then we also store a register $\ket{k} \in \mathcal H_k$, which stores the value of the level $k$ of $\mathcal T$ where the computation is in the process.

Next, there are $D$ quadruples of registers, for each level of the tree $\mathcal T$, which encode the state of the computation on an edge in $\mathcal T$.
The meanings of the registers of the $i$-th quadruple are, in this order:
\begin{itemize}
    \item The description (work) register $\ket{w^i} \in \mathcal H_w^i$: this will contain the data describing $v^i_t$.
    \item The child index register $\ket{j^i} \in \mathcal H^i_j$: an integer in $[\deg(v^i)]$ denoting the index of the child $v^i$ of its parent $v^{i-1}$.
    If $i > k$, then $\ket{j^i} = \ket{0}$.
    \item The total computation time register $\ket{f^i} \in \mathcal H^i_f$: the number of computation steps required to fully compute $v^i$ (equal to $t_{v^i}$).
    If $i > k$, then $\ket{f^i} = \ket{0}$.
    \item The moment of computation register $\ket{m^i} \in \mathcal H^i_m$: the number of computation steps already performed to compute $v^i$, between $0$ and $t_{v^i}$.
\end{itemize}
Therefore, the whole state of the algorithm belongs to the space $\mathcal H_p\otimes\mathcal H_k \otimes \mathcal H_s$, where
\[ \mathcal H_s = \bigotimes_{i=1}^D \mathcal H^i_w \otimes \mathcal H^i_j \otimes \mathcal H^i_f \otimes \mathcal H^i_m.\]

First of all, we assume that the number of qubits for all registers is known beforehand, therefore we require:
\begin{itemize}
    \item an upper bound $D_{\max}$ on $D$, to know how many quadruples to store;
    \item the register $\ket{j^i}$ contains $\lceil \log_2 n \rceil$ qubits; if additionally an upper bound $\Delta$ on the maximum degree in $\mathcal T$ is known, then only $\lceil\log_2 \Delta\rceil$ qubits are required;
    \item an upper bound $t_{\max}$ on $t_v$; then each of the registers $\ket{f^i}$ and $\ket{m^i}$ will consist of $\lceil{\log_2(t_{\max})}\rceil$ qubits;
    \item if a step of the transition uses at most $g$ elementary gates, $\ket{w^i}$ should have $\Oh(gt_{\max})$ qubits -- the exact amount depends on the specifics of the transition procedure in the application;
    \item $\ket{k}$ contains $\lceil \log_2 D \rceil$ qubits.
\end{itemize}

The overall state in $\mathcal H_s$ at any point is supported on the ``vertex'' states of the form
\begin{equation*}
     \ket{s(v,t)} = \left(\bigotimes_{i = 1}^{k(v)-1} \ket{v^i} \ket{j^i} \ket{t_{v^i}} \ket{t_{v^i}} \right) \otimes \Biggl(\ket{v_t} \ket{j^{k(v)}} \ket{t_v} \ket{t} \Biggr) \otimes \left( \bigotimes_{i=k(v)+1}^D \ket{0} \ket{0} \ket{0} \ket{0}\right),
\end{equation*}
for the corresponding vertex $s(v,t) \in \mathcal E$, where $k(v)$ is the level of $v$ in $\mathcal T$ ($0$-based starting with $r$).
The overall state of the algorithm between applications of $U$ can be written as
\[\ket{\ell(x) \Mod 2} \ket{k(x)} \ket{x},\]
where $x = s(v,t)$ corresponds to $v_t$ and $k(x) = k(v)$.
Here, $\ell(x)$ is the distance from the root to $x$ in the expanded tree $\mathcal E$.
The overall starting state of the algorithm is simply the all-zero state.

\paragraph{Step of the walk.}

We now proceed to describe the implementation of $U$.
An application of $U$ consists of alternating the application of $R_A$ and $R_B$.
Essentially, both $R_A$ and $R_B$ apply diffusions, where a single diffusion $D_x$ acts on the subspace $\mathcal H_x$.
The difference is that $R_A$ applies $D_x$ when $\ell(x)$ is even, and $R_B$ when it is odd.
The register $\ket{p}$ encodes this parity between applications of $U$.

Therefore, we can implement a single unitary $R$ that applies the diffusions $D_x$, and flip the parity register in between calls to $R$.
This way, we have to implement only one reflection operator $R$.
\begin{algorithm}
    \caption{Implementation of $U$.}
    \bigskip
    Repeat twice:
    \begin{enumerate}
        \item Apply $R$ to the state of the algorithm.
        \item Flip $\ket{p}$.
    \end{enumerate}
\end{algorithm}

Next, we describe the overall approach to implement $R$.
In general, we will make sure that the following invariant holds:
\begin{itemize}
    \item before the application of $R_A$, $p = 0$ iff the state in $\mathcal H_s$ corresponds to a vertex $x$ with $\ell(x)$ even;
    \item before the application of $R_B$, $p = 0$ iff the state in $\mathcal H_s$ corresponds to a vertex $x$ with $\ell(x)$ odd.
\end{itemize}
This is ensured by flipping $\ket p$ when we apply a transition, either forward or backward along an edge in $\mathcal E$; and by flipping $\ket{p}$ after each application of $R$.

Our general goal is to implement the diffusion operators $D_x$.
Fix one vertex $x$; by the invariant described above, $D_x$ must affect only the basis states $\ket{\rho_x}:=\ket{0} \ket{k(x)} \ket{x}$ and $\ket{\rho_y} := \ket{1} \ket{k(y)} \ket{y}$, where $y$ is a child of $x$.
The register $\ket{p}$ now denotes whether the vertex state is $\ket{x}$ or $\ket{y}$ for a child $y$ of $x$.

The way we will implement $D_x$ is as follows: first, we will uncompute a transition step for the states $\ket{1} \ket{k(y)} \ket{y}$, essentially taking one step back in the expanded tree, except for the registers $\ket{p}$ and $\ket{j^{k(y)}}$, which we do not change.
Let $\ket{\rho_x'}$ and $\ket{\rho_y'}$ be the states $\ket{\rho_x}$ and $\ket{\rho_y}$ after this operation, respectively.
Then the states $\ket{\rho_x'}$ and $\ket{\rho_y'}$ are identical except for the registers $\ket{p}$ and possibly $\ket{j^{k(y)}}$.

Then $p = 0$ for $x$ and $p=1$ for $y$.
If $k(x) = k(y)$, then $\ket{p}$ alone identifies the state $\ket{x}$ or $\ket{y}$ and we can implement $D_x$ as acting only on $\mathcal H_p$.
If $k(x) \neq k(y)$, then $k(y)=k(x)+1$ and $j^{k(x)+1}$ is equal to the index of the child $y$ of $x$, or $0$ for $x$.
In that case these two registers uniquely identify the basis state of $\mathcal H_x$ and we can implement $D_x$ as acting on these two registers only.
Thus, next we implement the required diffusion on $\mathcal H_p \otimes \mathcal H^{k(x)+1}_j$.
Afterwards, controlled on $p = 1$, we will apply the transition step in the forward direction to fix the overall state of the algorithm.

The algorithm acts differently depending on whether it needs to switch across the levels of $\mathcal T$.
This can happen in two cases: first, when $m^k=f^k$ and we need to compute a step ($p=0$); and second, when $t = 1$ and we need to uncompute a step ($p=1$).
This condition is first computed in a Boolean value $b$ in the beginning, and afterwards we uncompute it using the same condition, as it will stay invariant.

Therefore, the general structure is as follows:
\begin{algorithm}[H]
    \caption{Implementation of $R$.}
    \begin{enumerate}
        \item Compute $b$.
        \item \label{itm2} Controlled on $p = 1$, uncompute one step in the expanded tree.
        \item \label{itm3} Apply the corresponding operators $D_x$ to $\mathcal H_p \otimes \mathcal H^{k(x)+1}_j$.
        \item \label{itm4} Controlled on $p = 1$, compute one step in the expanded tree.
        \item Uncompute $b$.
    \end{enumerate}
\end{algorithm}

The computation of $b$ already has been explained above; we will explain the implementation of the other steps.
\begin{itemize}
    \item \textbf{Step \ref{itm2}: reversing a step of the computation.}
    Here we can just apply the operations of Step \ref{itm4} in reverse.
    
    \item \textbf{Step \ref{itm3}: performing the diffusion.} 
    There are two cases to consider, depending on the value of $b$.

    If $b=0$, then $x$ has a single child $y$, corresponding to applying the next computation step.
    Then to implement $D_x$, we need to generate the unit state $\ket{\psi_x} \propto \ket{x}+\ket{y}$.
    In our implementation, this means generating $\ket{\psi_x'} \propto \ket{0} +\ket{1} \in \mathcal H_p$ and $I-2\ket{\psi_x'}\bra{\psi_x'}$ is just $-X$ operation applied on $\ket{p}$.

    If $b=1$, then $x = u_{t_u}$ for some $u \in \mathcal T$, and 
    $\ket{\psi_x} \propto \sqrt{w_u}\ket{u} + \sum_{i \in [\deg(u)]} \sqrt{w_v} \ket{v}$, where $v$ is the $i$-th child of $u$.
    In our implementation, this means preparing the unit state $\ket{\psi_x'} \propto \sqrt{w_u}\ket{0}\ket{0} + \sum_{i \in [\deg(u)]} \sqrt{w_v} \ket{1}\ket{i} \in \mathcal H_p \otimes \mathcal H^{k(x)+1}_j$.
    Provided we can generate such a state $\ket{\psi_x'}$ from $\ket{0}$ with some circuit $S$, we can then implement $I - 2\ket{\psi_x'}\bra{\psi_x'}$ in a standard way as follows.
    First apply $S^{\dagger}$ to the given state, transforming any $\ket{\psi_x'}$ component to $\ket{0}$.
    Then flip the phase controlled on the state being $\ket{0}$, and apply $S$ back to the whole state.

    In general, $S$ must be implemented using the oracles $C$ and $O_t$.
    As the input, $C$ will take the description of the vertex in $\mathcal H_s$.
    In the context of query complexity, the complexity of $S$ doesn't matter as it never queries $A$.
    Time-wise, $S$ cannot be implemented efficiently without QRAM in the general case, for instance if the relevant $t_v$ are stored in memory.
    
    \item \textbf{Step \ref{itm4}: applying a step of the computation.}
    
    If $p = 1$, we have to take a step in the tree.

    If $b=1$, we need to start computing a vertex on the level $k+1$, whose index is currently stored in $\ket{j^{k+1}}$ after the diffusion.
    Hence we increment $\ket{k}$ by $1$ and add $t_{v^{k+1}}$ of the corresponding child to $f^{k+1}$ using the oracle $O_t$.

    Afterwards regardless of the value of $b$ we perform one computation step, applying $A$ on $\mathcal H^{k-1}_w \otimes \mathcal H^k_j \otimes \mathcal H^k_w \otimes \mathcal H^k_m$ (with the possibly updated value of $k$).
\end{itemize}

We can see that all of the steps in the algorithm can be implemented efficiently in the QRAM model with the exception of the circuit $S$.
Note that many times the algorithm uses $\ket{k}$ to perform operations on the $k$-th quadruple in $\mathcal H_s$: this is also efficient in the QRAM model.
Regarding the query complexity, each application of $R$ uses two queries to $A$, thus $U$ uses four queries to $A$ in total.
Therefore, we have the following statement.

\begin{claim} \label{thm:step}
    The unitary $U$ uses $4$ queries to each of $A$ and $S$, and can be implemented in $\Oh(\poly \log(t_{\max}, \Delta, D))$ additional gates in the QRAM model.
\end{claim}

\subsection{Unknown times} \label{sec:unknown}

By reducing this setting to the known times, we obtain the following result:

\begin{theorem} \label{thm:unknown}
    For a computation tree $\mathcal T$ with depth $D$ and unknown transition times $t_v$, there exists a bounded-error quantum algorithm that detects the presence of a marked vertex and uses $\Oh(\sqrt{TD\log t_{\max}})$ applications of $U$, where $T = \sum_v t_v^2$ and $t_{\max} = \max_v t_v$.
\end{theorem}

\begin{proof}
Suppose that a vertex $v$ needs an unknown time $t_v$ to be calculated.
Then we can run a version of \emph{exponential search}, running $A$ on $v$ for $1$, $2$, \ldots, $2^a$ steps until $v$ is computed, for increasing powers of $2$.
Here, $a$ is the largest integer such that $1+2+\ldots+2^a \geq t_v$.
Therefore, we have $a = \Oh(\log t_v)$.

If $A$ has calculated $v$ before the $2^a$-length block has expired, then we store an index of this step in a separate register, and for all remaining steps just apply the identity.
When uncomputing this block, this register determines when to start the uncomputation.

In total, we have replaced one vertex with $a+1$ vertices such that the sum of their squared times is equal to $1^2+2^2+\ldots+2^{2a} = \Oh(t_v^2)$.
The depth of the computation tree got multiplied by a factor of $\Oh(\log t_{\max})$.
Hence, by Theorem \ref{thm:known}, the total number of calls to $U$ is $\Oh(\sqrt{TD\log t_{\max}})$.
\end{proof}

We note that the implementation of $S$ in this case is quite simple.
Examine one vertex $v \in \mathcal T$ and let $h_v := 1+2+\ldots+2^a \geq t_v$ be the total number of vertices in $\mathcal E$ on the path on computing $v$.
\begin{itemize}
    \item First examine a vertex $x = s(v,t)$ for some $t < h_v$ and let $y = s(v,t+1)$.
    If $t = 1+2+\ldots+2^b$ for some $b < a$, then
    \[\ket{\psi_x} \propto \sqrt{2^b}\ket{x} +\sqrt{2^{b+1}}\ket{y} \propto \ket{x}+\sqrt{2}\ket{y}.\]
    Then in Step \ref{itm3} we only need to prepare a state proportional to $\ket{0} +\sqrt{2}\ket{1} \in \mathcal H_p$, which can be implemented using an appropriate rotation gate.
    For other $t \in [1,h_v)$, we have $\ket{\psi_x} \propto \ket{x}+\ket{y}$, so we need to generate a state proportional to $\ket{0}+\ket{1}$, which we can do with one Hadamard gate.
    \item If $x = s(v,h_v)$, then we need to generate
    \[\ket{\psi_x} \propto \sqrt{2^{h_u}}\ket{s(u,h_u)} + \sum_{v\,:\,u \underset{\mathcal T}{\rightarrow} v} \ket{s(v,1)}.\]
    Then in Step \ref{itm3} we have to prepare a state proportional to $\sqrt{2^{h_u}}\ket{0}\ket{0}+\ket{1}\sum_{i=1}^{\deg(u)} \ket{i}  \in \mathcal H_p \otimes \mathcal H^{k(x)+1}_j$.
    First we generate the state $\ket{\rho} \propto \sqrt{2^{h_u}} \ket{0} + \ket{1}$ using an appropriate rotation matrix.
    There are $\Oh(\log t_{\max})$ such rotations, which are known beforehand.
    Then we generate the uniform state $\ket{\chi} \propto \sum_{i=1}^{\deg(u)} \ket{i}$ on another $d = \lceil \log_2 \Delta \rceil$ qubits controlled on $\ket{\rho}$ being $\ket{1}$, using $\Oh(d) = \Oh(\log \Delta)$ gates as in, for example, \cite{SV24}.

    Similarly we can also generate the required $\ket{\psi_r}$ for the root vertex.
\end{itemize}

\section{Applications}

\subsection{Variable time search}

\begin{theorem} \label{thm:vts-algo}
    Let $x_1, \ldots, x_n \in \{0,1\}$ be variables where $x_i$ is computed by a (classical or quantum) procedure $A_i$ in $t_i$ steps, where $t_i$ is unknown beforehand, and let $T = \sum_i t_i^2$.
    Suppose that we are given upper bound estimates $T'$ and $t'$ such that $T \leq T' = \Oh(T)$ and $t_{\max} \leq t' = \poly(t_{\max})$.
    Then there is a bounded-error quantum algorithm that detects a $j$ such that $x_j=1$ using $\Oh(\sqrt{T \log t_{\max}})$ queries to steps of the procedures $A_i$.
\end{theorem}

\begin{proof}
    We can represent this problem as a star tree with $n+1$ vertices, where the computation time of the $i$-th child is $t_i$ and the child is marked if $x_i = 1$, see Figure \ref{fig:vts}.
    Then the required immediately follows from Theorem \ref{thm:unknown}, since $D = 1$.
    \begin{figure}[H]
    \begin{center}
        \begin{tikzpicture}[scale=0.9,
          level 1/.style={sibling distance=1.5cm, level distance=2cm},
          vertex/.style={circle, fill=black, inner sep=1.5pt}
        ]
          \node[vertex] {}
            child {node[vertex] {} edge from parent node[xshift=-0.7cm, yshift=-0.1cm] {$t_1$}}
            child {node[vertex] {} edge from parent node[xshift=-0.4cm, yshift=-0.1cm] {$t_2$}}
            child {node[vertex] {} edge from parent node[right, yshift=-0.1cm] {$\ldots$}}
            child {node[vertex] {} edge from parent node[right, yshift=-0.1cm] {}}
            child {node[vertex] {} edge from parent node[xshift=0.7cm, yshift=-0.1cm] {$t_n$}};
        \end{tikzpicture}
        \caption{The computation tree of variable time search.}
        \label{fig:vts}\qedhere
    \end{center}
    \end{figure}
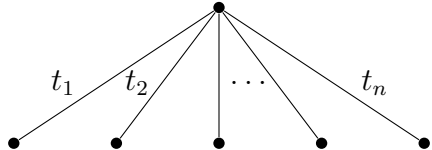
\end{proof}

\paragraph{Optimal query complexity.}

Next we adapt this algorithm to achieve optimal asymptotic query complexity.

\begin{theorem} \label{thm:vts}
    Suppose that we are given upper bound estimates $T'$ and $t'$ such that $T \leq T' = \Oh(T)$ and $t_{\max} \leq t' = \poly(t_{\max})$.
    Then the quantum query complexity of VTS with unknown times is 
    \[\Theta(\sqrt{T \log \min(n, t_{\max})}).\]
\end{theorem}

\begin{proof}
    First we prove the lower bound.
    It has been proven in \cite{AKV23} that $\Omega( \sqrt{T \log T })$ queries are necessary.
    The parameters of their construction are $n = \Theta(T/\log T)$ and $t_{\max} = \Theta(\sqrt[4]{T/\log T})$.
    Hence both $\log n$ and $\log t_{\max}$ are $\Theta(\log T)$ and the required follows.

    For the upper bound, first we have an $\Oh(\sqrt{T \log t_{\max}})$ algorithm from Theorem \ref{thm:vts-algo}.
    It remains to show a query algorithm with complexity $\Oh(\sqrt{T \log n})$, for example, when $n < t_{\max}$.
    For each item $i$, treat each $k = t_{\max}/n$ consecutive steps of the computation as just one step (if the last block contains less than $k$ original steps, perform the identity operation for the remaining ones).
    Then we have $t_i' = \Oh(t_i/k)$ for all $i$ and $T' = \sum_i t_i'^2 = \Oh(T/k^2)$.
    By applying Theorem \ref{thm:vts-algo}, we get query complexity
    \[\Oh(k\sqrt{T' \log t'_{\max}}) = \Oh\left(\frac{t_{\max}}{n}\sqrt{\frac{Tn^2}{t_{\max}^2}\log n}\right) = \Oh(\sqrt{T \log n}).\qedhere\]
\end{proof}

\paragraph{Weighing choices.}

A discussion is fitting at this point about the choice of weights for the search in a tree with unknown variable times (also see Table 1 in \cite{Jef25}). 
It is interesting to compare our choice of weights to the ones of \cite{Jef25,Cor25}, since they also deal with exploring a path of unknown length and use similar ideas.
Consider the computation tree of VTS, where the path to $x_i$ contains $t_i$ vertices.
Suppose that the $j$-th vertex on the $i$-th path has weight $w_{i,j} := w_{s(x_i,j)}$.
Informally, it follows from the electrical network framework \cite{Bel13} that the complexity of the quantum walk is $\Oh(\sqrt{\mathcal R \mathcal W})$, where $\mathcal R = \max_{i=1}^n R_i$ for $R_i = \sum_{j=1}^{t_i} 1/w_{i,j}$ and $\mathcal W = \sum_{i=1}^n W_i$ for $W_i = \sum_{j=1}^{t_i} w_{i,j}$.

Consider one of these paths and the effect of the choice of weights on the complexity.
For simplicity, let its length be $t$, the weights $w_1$, $\ldots$, $w_t$, and denote corresponding $R_i$ and $W_i$ by $R$ and $W$.
Examine the following four options:
\begin{enumerate}
    \item \label{itm:opt1} If we know $t$ beforehand, then we can pick $w_i = t$ for all $i$.
    In this case $R = 1$ and $W = t^2$.
    This corresponds to the choice of weights in Theorem \ref{thm:known}.
    \item \label{itm:opt2} For unknown $t$, the choice of weights in Theorem \ref{thm:unknown} is $w_1 = 1$, $w_2, w_3 = 2$, $\ldots$ in blocks of length $2^i$, in which case $R \sim \log_2 t$ and $W \sim \frac{4}{3} \cdot t^2$.
    \item \label{itm:opt3} For unknown $t$, the choice of weights in \cite{Jef25,Cor25} is $w_i = i$, then $R \sim \ln t$ and $W \sim \frac 1 2 \cdot t^2$.
    \item \label{itm:opt4} For unknown $t$, we can also pick $w_i = 1$ for all $i$.
    This results in $R = W = t$.    
\end{enumerate}

Option \ref{itm:opt3} is better than Option \ref{itm:opt2} by constant factors; Option \ref{itm:opt2} is more gate-friendly, because it is easy to construct a small circuit that performs the relevant reflections that can be implemented by just a constant number of standard gates as described in Section \ref{sec:unknown}.
Then we only need to check whether the number of the current step is equal to $2^j$, and controlled on this condition, apply the reflection above.
This requires no use of QRAM, as is needed in Option \ref{itm:opt3} to fetch the gate to perform the required step-specific reflection.
Of course, whether we need QRAM for the whole algorithm still depends on whether we need QRAM to implement other oracles.
Option \ref{itm:opt2} also conceptually shows how the standard classical exponential search procedure is ``quantized'' in our algorithm. 

Option \ref{itm:opt1} results in an algorithm with complexity $\Oh(\sqrt{\sum_i t_i^2})$, which corresponds to VTS with known times.
Options \ref{itm:opt2} and \ref{itm:opt3} result in complexity $\Oh(\sqrt{\sum_i t_i^2 \log t_{\max}})$, corresponding to VTS with unknown times.
Option \ref{itm:opt4} is curious, as it has complexity $\Oh(\sqrt{\sum_i t_i \cdot t_{\max}})$.
In fact, it can be advantageous if $n = 1$, in which case the complexity is just $\Oh(t)$, compared to $\Oh(t \log t)$ with Options \ref{itm:opt2} and \ref{itm:opt3}.
Amusingly enough, this is just a classical algorithm of time $t$ simulated by a quantum walk!
In fact, the quantum walk becomes fully classical, as the reflections used at the vertices of the path are just $-X$ operations, propagating an absolute amplitude of $1$ from the root to the end of the path.

\subsection{Divide \& conquer}

Recall that we consider divide \& conquer algorithms with input size $n$, which partitions a problem into $a$ subproblems of size $n/b$, for some constants $a$ and $b$.
We assume that a subproblem of size $n$ can be partitioned by an exact classical or quantum procedure in $T_{\aux}(n)$ steps.
We assume that the function $T_{\aux}$ is known, since provable running times are achieved with such an assumption. 
Then we defined $T_{\Q}(n)$ by a recurrence
\[ T_{\Q}(n) \leq \sqrt{a}T_{\Q}(n/b) + T_{\aux}(n).\]

\begin{theorem} \label{thm:dnq}
    There is a bounded-error quantum algorithm for divide \& conquer with disjunctive combining with query complexity $\Oh(T_{\Q}(n) \sqrt{\log n})$ and time complexity $\widetilde \Oh(T_{\Q}(n))$ in the QRAM model.
\end{theorem}

\begin{proof}
    For simplicity, we assume that \emph{each} of $a$ subproblems is computed by a procedure with running time $T_{\aux}(n)$; since $a$ is constant, this will not affect the asymptotic complexity.
    If the combination of the recursive calls is performed by \textsc{Or}, then the divide \& conquer algorithm can be represented as a complete $a$-ary tree $\mathcal T$ of depth $D = \Oh(\log_b n)$, see Figure \ref{fig:dnq}.

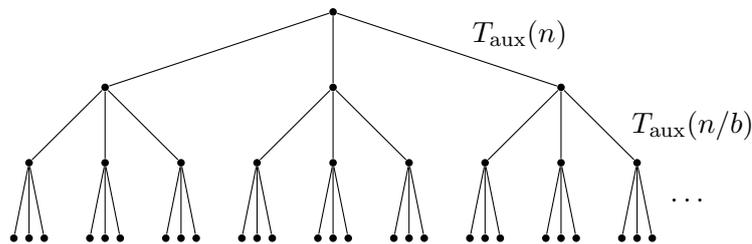
\begin{figure}[H]
\small
\begin{center}
\begin{tikzpicture}[
  vertex/.style={circle, fill=black, inner sep=1pt},
  level 1/.style={sibling distance=3cm, level distance=1cm},
  level 2/.style={sibling distance=1cm, level distance=1cm},
  level 3/.style={sibling distance=0.2cm, level distance=1cm},
  edge from parent/.style={draw, black}
]

\node[vertex] {}
  child {node[vertex] {}
    child {node[vertex] {}
        child {node[vertex] {}
        edge from parent node[left=.2cm] {~~~~~~~~~}}
        child {node[vertex] {}}
        child {node[vertex] {}}
    }
    child {node[vertex] {}
        child {node[vertex] {}}
        child {node[vertex] {}}
        child {node[vertex] {}}
    }
    child {node[vertex] {}
        child {node[vertex] {}}
        child {node[vertex] {}}
        child {node[vertex] {}}
    }
  }
  child {node[vertex] {}
    child {node[vertex] {}
        child {node[vertex] {}}
        child {node[vertex] {}}
        child {node[vertex] {}}
    }
    child {node[vertex] {}
        child {node[vertex] {}}
        child {node[vertex] {}}
        child {node[vertex] {}}
    }
    child {node[vertex] {}
        child {node[vertex] {}}
        child {node[vertex] {}}
        child {node[vertex] {}}
    }
  }
  child {node[vertex] {}
    child {node[vertex] {}
        child {node[vertex] {}}
        child {node[vertex] {}}
        child {node[vertex] {}}
    }
    child {node[vertex] {}
        child {node[vertex] {}}
        child {node[vertex] {}}
        child {node[vertex] {}}
    }
    child {node[vertex] {}
        child {node[vertex] {}}
        child {node[vertex] {}}
        child {node[vertex] {}
               edge from parent node[right=.2cm] {$\ldots$}}
        edge from parent node[right=.3cm] {$T_{\aux}(n/b)$}
    }
    edge from parent node[above=.2cm, right=.2cm] {$T_{\aux}(n)$}
  };

\end{tikzpicture}
\end{center}
    \caption{The computation tree of divide \& conquer.}
    \label{fig:dnq}
\end{figure}

A vertex $v$ on the $i$-th level of this tree has computation time $t_v = T_{\aux}(n/b^i)$.
Therefore,
\[T = \sum_v t_v^2 = \sum_{i=0}^D a^i \cdot T_{\aux}^2( n/b^i).\]
Since $\sqrt{p^2+q^2} < p + q$ for positive $p$ and $q$, this implies
\[\sqrt{T} < \sum_{i=0}^D \sqrt{a^i}\cdot T_{\aux}(n/b^i).\]
We can conclude that $\sqrt{T}$ satisfies the recurrence $T_{\Q}(n)$.
By then applying then Theorem \ref{thm:known} to this tree, we obtain a quantum algorithm with query complexity
\[\Oh(\sqrt{TD}) = \Oh(T_{\Q}(n)\sqrt{\log n}).\]
Finally, the time-efficient version follows from Claim \ref{thm:step} by noticing that the oracle $S$ can be implemented efficiently, as only uniform superpositions over the children need to be prepared because of the symmetry.
\end{proof}

\paragraph{Point on 3 lines.}

Here we give a new application of the quantum divide \& conquer for the problem of determining an intersection of three lines, improving the $\Oh(n^{1+\oh(1)})$ complexity of the previous quantum algorithm \cite{AL20}.

\begin{theorem} \label{thm:po3l}
    There is a bounded-error quantum algorithm that solves \textsc{Point-On-3-Lines} with time complexity $\widetilde \Oh(n)$ in the QRAM model.
\end{theorem}

\begin{proof}
    First we start with the definition of the relevant planar data structure.
    For a set of $n$ lines $S$, an $(1/r)$-cutting is a set of disjoint simplices (cells) that cover the whole plane, where each simplex is intersected by at most $\Oh(n/r)$ lines of $S$.
    There exist deterministic algorithms that construct an $(1/r)$-cutting of size $\Oh(r^2)$ in time $\Oh(nr)$ \cite{Cha93,CT16}. 

    Now consider the following classical algorithm for \textsc{Point-On-3-Lines}: pick a constant $r$ and construct an $(1/r)$-cutting for the given set of lines.
    Then for each of the $\Oh(r^2)$ cells, find all lines that intersect it and call the algorithm recursively.
    When a cell contains a constant number of lines, check in constant time whether there are three concurrent lines by examining all triples.
    
    On the $i$-th level of recursion ($0$-based), the algorithm consumes $\Oh((n/r^i) r^2) = \Oh(n/r^i)$ time and the $i$-th level contains $\Oh(r^{2(i+1)})$ recursive calls.
    The depth of the computation tree is only $D = \Oh(\log_r n)$.
    By applying Theorem \ref{thm:dnq}, we get complexity
    \[\widetilde \Oh\left(\sqrt{\sum_{i=0}^D r^{2(i+1)} \left(\frac{n}{r^i}\right)^2}\right) = \widetilde \Oh(n). \qedhere\]
\end{proof}

\subsection{Bomb query algorithms}

Finally, we briefly show how quantum algorithms based on guessing decision trees, also known as quantum bomb query algorithms, can be seen in our framework.

First we start with the definitions.
Consider a Boolean function $f : \{0,1\}^n \to \{0,1\}$ and a decision tree $A$ that computes $f$.
Suppose that additionally at each node $v$ of $A$ we have a bit which is a guess for the query result at this node.
Suppose that on an input $x \in \{0,1\}^n$, $A$ makes $T$ queries and $G$ incorrect guesses.
Then there exists a bounded-error quantum algorithm that computes $f(x)$ in $\Oh(\sqrt{TG})$ queries, as shown in \cite{LL15,CMP22}.

The idea for this quantum algorithm is as follows: first, generate a guess string $g_1\ldots g_m$, assuming the guesses at all nodes are correct.
Then run a variant of Grover's search that finds the first position $j$ where $g_j$ differs from the actual query at that node on $x$ \cite{Kot14}, and then repeat the whole algorithm starting from that node.
If the position $j$ at the runs of Grover's search is given by $d_1$, $\ldots$, $d_G$, then this algorithm performs $t_i = \Oh(\sqrt{d_i})$ queries to find the $i$-th wrong guess.
For now, assume that each such procedure has zero probability of error.

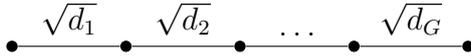
\begin{figure}[H]
\begin{center}
\begin{tikzpicture}[scale=1, vertex/.style={circle, fill=black, inner sep=1.5pt}]
  \node[vertex] (v1) at (0,0) {};
  \node[vertex] (v2) at (1.5,0) {};
  \node[vertex] (v3) at (3,0) {};
  \node[vertex] (v4) at (4.5,0) {};
  \node[vertex] (v5) at (6,0) {};

  \draw (v1) -- node[above] {$\sqrt{d_1}$} (v2);
  \draw (v2) -- node[above] {$\sqrt{d_2}$} (v3);
  \draw (v3) -- node[above] {$\ldots$} (v4);
  \draw (v4) -- node[above] {$\sqrt{d_G}$} (v5);
\end{tikzpicture}
\caption{The computation tree of the quantum bomb query algorithm.}
\end{center}
\end{figure}

The computation tree of this algorithm then is just a line of depth $D=G$, with computation cost at the $i$-th level being an unknown time $t_i$.
Since $\sum_{i=1}^G t_i^2 = \Oh(T)$, by applying Theorem \ref{thm:unknown}, we obtain a quantum query algorithm with complexity
\[\Oh\left(\sqrt{\sum_{i=1}^D t_i^2 D \log t_{\max}}\right) = \Oh\left(\sqrt{TG\log T}\right).\]

We can of course apply the Cauchy-Schwarz inequality to estimate that
\[\sum_{i=1}^G t_i = \Oh\left(\sum_{i=1}^G \sqrt{d_i}\right) = \Oh\left(\sqrt{\sum_{i=1}^G \left(\sqrt{d_i}\right)^2G}\right) = \Oh(\sqrt{TG}),\]
obtaining the same complexity (and even better by a square root of a logarithmic factor).
Our approach simply gives a conceptual viewpoint at how the complexity $\widetilde \Oh(\sqrt{TG})$ fits in our framework.

Moreover, we assumed that first element Grover's searches have no error, which is not true; if we need to find the leftmost of  multiple marked elements, then there is a constant probability of error.
Span programs can be used to avoid an additional logarithmic factor in the complexity from error reduction \cite{LL15,CMP22}.
For simplicity's sake, in this paper we have considered only exact computation procedures, which would require further logarithmic factors in the complexity to reduce their error.
Perhaps the recent breakthrough of \emph{transducers} \cite{BJY24} that allows to compose quantum algorithms without additional complexity for error reduction can be incorporated into our algorithm. 

\section{Acknowledgements}

We thank Andris Ambainis for introducing us to the problem and for the observation that led to Theorem \ref{thm:vts}.
We also thank Aleksandrs Belovs and Krišjānis Prūsis for helpful comments.

Supported by QOPT (QuantERA ERA-NET Cofund) and by Accenture Baltics conducted in collaboration with a dispersed Accenture team under an “Agreement on quantum computing use case research”.

\printbibliography

\end{document}